\documentclass[pra,onecolumn,groupedaddress,superscriptaddress,amsmath,amssymb,a4paper]{revtex4}
\usepackage{geometry}
\usepackage{hyperref}
\usepackage{amssymb}
\usepackage{amsmath}
\usepackage{amsthm}
\usepackage{amsbsy}
\usepackage{physics}
\usepackage{mathrsfs}
\usepackage{qcircuit}
\usepackage{braket}
\usepackage{tikz}

\usepackage{bm}
\usepackage{xcolor}

\definecolor{urlcolor}{HTML}{000088} 
\definecolor{linkcolor}{HTML}{008800} 
\definecolor{citecolor}{HTML}{008800} 

\hypersetup{pdfstartview=FitH,  linkcolor=linkcolor,urlcolor=urlcolor, citecolor   = citecolor, colorlinks=true}

\usetikzlibrary{positioning,arrows}
\usepackage[T1]{fontenc}

\newtheorem{proposition}{Proposition}
\theoremstyle{definition}

\newtheorem{example}{Example}
\newtheorem{theorem}{Theorem}

\newcommand{\I}{\mathcal{I}}
\newcommand{\Id}{\mathcal{I}^\dagger}

\begin{document}

\title{Quantum state tomography via sequential uses of the same informationally incomplete measuring apparatus}

\author{\firstname{V.\,A.}~\surname{Zhuravlev}}
 \affiliation{Moscow Institute of Physics
and Technology, Institutskii Per. 9, Dolgoprudny 141700, Russia}

\author{\firstname{S.\,N.}~\surname{Filippov}}
 \affiliation{Moscow Institute of Physics
and Technology, Institutskii Per. 9, Dolgoprudny 141700, Russia}
 \affiliation{Steklov
Mathematical Institute of Russian Academy of Sciences, Gubkina St.
8, Moscow 119991, Russia}

\begin {abstract}{State of a $d$-dimensional quantum system can only be inferred by performing an informationally complete measurement with $m\geqslant d^2$ outcomes. However, an experimentally accessible measurement can be informationally incomplete. Here we show that a \emph{single} informationally \emph{incomplete} measuring apparatus is still able to provide all the information about the quantum system if applied several times in a row. We derive a necessary and sufficient condition for such a measuring apparatus and give illustrative examples for qubits, qutrits, general $d$-level systems, and composite systems of $n$ qubits, where such a measuring apparatus exists. We show that projective measurements and L\"{u}ders measurements with 2 outcomes are useless in the considered scenario.}
\end{abstract}


  \maketitle

\section{Introduction}

The results of measurements on quantum systems are naturally probabilistic~\cite{holevo-st}. The distribution of measurement outcomes is defined by a positive operator-valued measure (POVM). There are two distinctive classes of POVMs: (i) informationally complete and (ii) informationaly incomplete. The former class establishes a one-to-one correspondence between the system density operator and the distribution of measurement outcomes~\cite{heinosaari-ziman}. The latter class fails in providing such an injective function: the distribution of measurement outcomes does not define the system density operator uniquely. 

Despite a POVM characterizes the statistics of outcomes for a particular measurement device, it does not describe the system transformation caused by the measurement. The measurement-induced system transformation is described by a quantum instrument that assigns a completely positive trace-nonincreasing map to each measurement outcome~~\cite{davies-lewis-1970,holevo-book,heinosaari-ziman}. This is the formalism of quantum instruments that enables one to deal with the system state after the measurement and subject it to further transformations, e.g., subsequent measurements. Sequential measurements are exactly the scenario we study in this paper. This scenario is exceptionally productive when none of the measurements in sequence is informationally complete because a multidisribution of outcomes for all measurements may, nevertheless, contain all the desired information~\cite{carmeli-2011,carmeli-2012,kalev-2012,lorenzo-2013}. A typical scenario is to perform a so-called weak measurement and then apply a projective measurement depending on the outcome observed~\cite{lundeen-2012}. Refs.~\cite{thekkadath-2016,calderaro-2018} report the state reconstruction schemes with several uses of indirect measurements before the projective one.

The goal of this paper is to consider an experimentally relevant situation of a \emph{single} measurement device available, which is informationally incomplete. The authors of Ref.~\cite{haapasalo-2016} noticed that repeated measurements with a single device could be used to correct the inherent noise of an unsharp observable. It turns out, however, that several applications of the same device can also result in the informationally complete statistics of outcomes. We provide specific restrictions on the informationally incomplete measuring apparatus under which the statistics of outcomes for $n$ uses of the apparatus enables a precise quantum state tomography. This research line continues a discussion of the quantum state dynamics under repeated measurements started in Refs.~\cite{haapasalo-2016,luchnikov-2017}.

The paper is organized as follows. In Section~\ref{section-instruments}, we provide a brief overview of quantum instruments and POVMs. In Section~\ref{section-general}, we derive a necessary and sufficient condition for the informational completeness of $n$ uses of the same informationally incomplete measuring apparatus. In Section~\ref{section-qubits}, we present an informationally incomplete measuring apparatus with 2 outcomes such that 2 sequential uses of this apparatus enable a precise state reconstruction. In Section~\ref{section-qudits}, we generalize the results of Section~\ref{section-qubits} to the case of $d$-dimensional quantum systems. In Section~\ref{section-n-qubits}, multiqubit systems are considered. In Section~\ref{section-conclusions}, brief conclusions are given. 

\section{Quantum instruments and non-destructive measurements} \label{section-instruments}

In this section, we briefly review density operators, POVMs, and quantum instruments as some requisite notions for further analysis of sequential quantum measurements. 

We consider non-trivial finite dimensional Hilbert spaces ${\cal H}_d$, $d = {\rm dim}{\cal H}_d > 1$. The state of a quantum system is represented by a Hermitian positive-semidefinite operator $\rho$ with the unit trace. Set of all states is denoted by $\mathcal{S}({\cal H}_d)$. The set $\mathcal{S}({\cal H}_d)$ is convex.

An effect $\mathsf{E}: \mathcal{S}({\cal H}_d) \to [0, 1]$ is an affine mapping from $\mathcal{S}({\cal H}_d)$ to $[0, 1]$ such that
\begin{equation*}
\mathsf{E}\left(\sum_i \lambda_i \rho_i\right) = \sum_i \lambda_i \mathsf{E}(\rho_i), \quad \lambda_i \geqslant 0, \quad \sum_i \lambda_i = 1,
\end{equation*}

\noindent which defines the mapping for a mixture of quantum states. Every effect has an associated positive-semidefinite bounded operator $E$ on ${\cal H}_d$ such that $O\leq E \leq I$, where $I$ is the identity operator and $O$ is the zero operator. The relation 
\begin{equation} \label{Born-rule}
\mathsf{E}(\rho) = \tr \left( \rho E \right), \quad \forall \rho \in \mathcal{S}({\cal H}_d)
\end{equation}

\noindent uniquely defines the operator $E$, which we will refer to as an effect too. By ${\cal E}({\cal H}_d)$ denote a set of effects for a $d$-dimensional quantum system.

Let $\Omega$ be a set of elementary outcomes in some physical experiment on a $d$-dimensional quantum system and $\mathcal{F}$ be a $\sigma$-algebra of events. A positive operator-valued measure (POVM) is a mapping ${\sf A}: {\cal F} \to {\cal E}({\cal H}_d)$ such that ${\sf A}(\emptyset) = O$, ${\sf A}(\Omega) = I$, and ${\sf A}(\cup_i X_i) = \sum_i {\sf A}(X_i)$ for any sequence $\{X_i\}$ of disjoint sets in ${\cal F}$~\cite{heinosaari-ziman}. In what follows, we consider a finite set $\Omega=\{x_k\}_{k=1}^m$, which corresponds to an $m$-outcome measurement. For a subset $X \subset \Omega$ we have ${\sf A}(X) = \sum_{x_k \in X} A(x_k)$. The effect $E_k := {\sf A}(x_k)$ defines the probability $p_k = {\rm tr}(\rho E_k)$ for observing a particular outcome $x_k$ provided the system state is described by the density operator $\rho$. The total probability $\sum_{k=1}^m p_k = 1$ regardless of the density operator $\rho$ because the effects $E_k$ satisfy the relation
\begin{equation} \label{identity-resolution}
\sum \limits_k E_k = I.
\end{equation}

A POVM is called informationally complete if it realises an injective mapping from $\mathcal{S}({\cal H}_d)$ to the set of probability distributions on $\Omega$. An informationally complete POVM enables a precise state tomography of an unknown quantum state. Particular reconstruction schemes are presented, e.g., in~\cite{bogdanov-2010,filippov-2011}. For a POVM to be informationally complete, the number of outcomes $m$ has to satisfy the relation $m \geqslant d^2$ as there are $d^2$ linearly independent operators acting on ${\cal H}_d$ and ${\rm dim} \, {\rm Span}({\cal S}({\cal H}_d)) = d^2$ (see, e.g.,~\cite{heinosaari-ziman}). Therefore, measurements with $m<d^2$ outcomes are informationally incomplete. Restriction on the number of outcomes has operational meaning~\cite{fghl-2019}. Note, however, that some POVMs (including informationally complete ones) can be simulated by other POVMs with less number of outcomes~\cite{fhl-2018}.

To describe the state transformation caused by a measuring apparatus we need to review the concept of a quantum operation. As we deal with a finite-dimensional Hilbert space ${\cal H}_d$, for our purposes it suffices to consider the space ${\cal L}({\cal H}_d)$ of linear operators acting on ${\cal H}_d$. Then we define a quantum operation $\Phi: {\cal L}({\cal H}_d) \to {\cal L}({\cal H}_d)$ as a linear, completely positive, and trace nonincreasing map. Complete positivity of $\Phi$ means the map $\Phi \otimes {\rm Id}$ is positive for all identity transformations ${\rm Id}: {\cal L}({\cal H}_n) \to {\cal L}({\cal H}_n)$, $n \in \mathbb{N}$, i.e., $\Phi \otimes {\rm Id} [X] \geqslant O$ for all $O \leqslant X=X^{\dag} \in {\cal L}({\cal H}_d \otimes {\cal H}_n)$. Physical meaning of complete positivity is discussed, e.g., in~\cite{filippov-jms-2019}. Trace nonincreasing property means ${\rm tr}(\Phi[X]) \leqslant {\rm tr}(X)$ for all $O \leqslant X=X^{\dag} \in {\cal L}({\cal H}_d)$. By ${\cal O}({\cal H}_d)$ denote the set of quantum operations $\Phi: {\cal L}({\cal H}_d) \to {\cal L}({\cal H}_d)$.

An important concept in quantum information theory is an \emph{instrument} that assigns a quantum operation to any outcome from the outcome space $(\Omega, \mathcal{F})$. A mapping ${\cal I}: \mathcal{F} \to {\cal O}({\cal H}_d)$ is called an instrument if ${\rm tr} \big( {\cal I}(\Omega)[\rho] \big) = {\rm tr}(\rho) = 1$, ${\rm tr} \big( {\cal I}(\emptyset)[\rho] \big) = 0$, ${\rm tr} \big( {\cal I}(\cup_i X_i)[\rho] \big) = \sum_i {\rm tr} \big( {\cal I}(X_i)[\rho] \big)$ for all $\rho \in {\cal S}({\cal H}_d)$ and any sequence of mutually disjoint sets $\{X_i\}$, $X_i \subset {\cal F}$. The first condition implies ${\cal I}(\Omega)$ is not only completely positive but also trace preserving, i.e., a quantum channel. Physically, observation of an outcome $x_j$ while performing a discrete measurement on a quantum system in the state $\rho$ results in the disturbance of the state $\rho \to {\cal I}(x_j)[\rho]$ described by the quantum operation ${\cal I}(x_j)$. The conditional (normalized) output state reads $\frac{ {\cal I}(x_j)[\rho] }{{\rm tr}\big( {\cal I}(x_j)[\rho] \big)}$. In what follows, we consider discrete measurements and use brief notations $j$ and ${\cal I}_j$  to refer to $x_j$ and ${\cal I}(x_j)$, respectively. 

The relation between the instrument and the corresponding POVM is straightforward. The probability $p_j$ to get outcome $j$ for the input state $\rho$ equals
\begin{equation} \label{prob-instrument}
p_j = {\rm tr}({\cal I}_j [\rho])
\end{equation}

\noindent and the total probability to get any of the outcomes $j=1,\ldots,m$ equals
\begin{equation*}
\sum_{j=1}^m p_j = {\rm tr} \left( \sum_{j=1}^m {\cal I}_j [\rho] \right) = {\rm tr}({\cal I}(\Omega) [\rho]) = {\rm tr}[\rho] = 1.
\end{equation*}

\noindent Combining~\eqref{Born-rule} and~\eqref{prob-instrument}, we readily get 
\begin{equation} \label{effect-through-instrument}
E_j = {\cal I}_j^{\dag}(I),
\end{equation}

\noindent where $\Phi^{\dag}$ denotes a dual map with respect to $\Phi$, i.e., ${\rm tr}\big( \Phi^{\dag}[X] Y \big) = {\rm tr}\big( X \Phi[Y] \big)$ for all $X,Y \in {\cal L}({\cal H}_d)$. The condition $E_j \geqslant O$ is fulfilled because the dual map ${\cal I}_j^{\dag}$ is completely positive. The condition~\eqref{identity-resolution} is fulfilled because ${\cal I}(\Omega)$ is trace preserving and, consequently, $\big( {\cal I}(\Omega) \big)^{\dag}$ is unital, i.e., $\big( {\cal I}(\Omega) \big)^{\dag}[I] = I$.

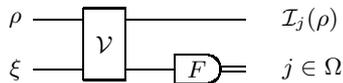
\begin{figure}
$\;$ \Qcircuit @C=1em @R=1em {
\lstick{\rho}& \qw & \multigate{1}{\mathcal{V}} & \qw & \qw & \qw &\rstick{\mathcal{I}_{j} (\rho)}\\ \lstick{\xi} & \qw & \ghost{\mathcal{V}} & \qw & \measureD{F} & \cw & \rstick{j \in \Omega} \\}  
\caption{\label{figure1} Schematic of a non-destructive measurement. $\mathcal{V}$ is a channel that couples a system in state $\rho$ and a probe in state $\xi$, $F$ is a POVM for the probe.}
\end{figure}

The mathematical formalism of a quantum instrument is tightly connected with a physical measurement model also referred to as a measuring process~\cite{ozawa-1984}, which describes an interaction between the system and a probe followed by a conventional measurement of the probe described by POVM (see Fig.~\ref{figure1}). The measurement model in Fig.~\ref{figure1} illustrates how the system in question interacts with a probe (their common evolution is described by a quantum channel ${\cal V}$), and the probe is measured afterwards. The resulting system transformation reads
\begin{equation} \label{instrument-from-measurement-model}
\mathcal{I}_j[\rho] = \tr_p \left( \mathcal{V} (\rho \otimes \xi) \left[ I \otimes F_j \right] \right),
\end{equation}

\noindent where $\tr_p$ denotes a partial trace over the probe's degrees of freedom. One can readily verify that formula~\eqref{instrument-from-measurement-model} defines a legitimate instrument. Remarkably, the inverse relation always holds true too: for any instrument ${\cal I}$ there exist a probe space, a probe initial state $\xi$, a channel ${\cal V}$, and a POVM $j \to F_j$ for the probe such that the relation~\eqref{instrument-from-measurement-model}, see, e.g.,~\cite{heinosaari-ziman}.

Importantly, the presented measurement model is non-destructive in a sense that the system is still available for further analysis after the measurement. In the next section, we consider sequential applications of the same non-destructive measurement apparatus to the system.

\section{Sequential uses of the same non-destructive measuring apparatus and informational completeness} \label{section-general}

We consider the measuring apparatus as a black box, whose input is a system in some state $\rho$ and whose output consists of two entities: a classical discrete outcome $j \in \{1, \ldots,m\}$ and a quantum system in the state $p_j^{-1} {\cal I}_j[\rho]$, see Fig.~\ref{figure2}(a). The measuring apparatus can therefore be applied again, now to the quantum outcome of its first use. A classical outcome for the second use, however, does not have to coincide with that for the first use. This enables extracting more information about the original system. The procedure can be continued in the same way arbitrarily many times. Let $N$ be the total number of the apparatus uses. Then we end up with a ``measurement tree diagram'' of depth $N$ describing all possible state transformations, see Fig.~\ref{figure2}(b). Interestingly, the problem of whether the measurement tree has an outcome that never occurs is undecidable~\cite{eisert-2012}.

Collecting classical outcomes for $N$ uses of the measuring apparatus, we get a multiindex $j_1 j_2 \ldots j_N$. The probability of observing a particular multiindex $j_1 j_2 \ldots j_N$ equals
\begin{equation} \label{probability-multiindex}
p_{j_1 j_2 \ldots j_N} = {\rm tr} \bigg( {\cal I}_{j_N} \Big[ \ldots {\cal I}_{j_2} \big[ {\cal I}_{j_1}[\rho] \big] \Big] \bigg). 
\end{equation}
\begin{figure}
\begin{minipage}[b]{0.3\textwidth}
     \begin{tikzpicture}[thick,
                    node distance=2.2cm,
                    text height=1.5ex,
                    text depth=.25ex,
                    auto, 
                    squarednode/.style={rectangle, draw=red!60,
                    fill=red!5, very thick, minimum size=5mm},
                    writeonarrow/.style = {pos=0.5,sloped,above}]
        \node [squarednode] (root) {$\mathcal{I}$};
        \node [right of=root] (empty) {};
        \node [left of=root] (emptyL) {};
        \node [below of=empty] (empty1) {};
        \node [above of=empty] (empty2) {};
        
        \path [->, draw] (emptyL) edge node [writeonarrow]  {$\rho$} (root);
        \path[->, draw] (root) edge node  [writeonarrow] {$\mathcal{I}_j[\rho]$} (empty);
        \path[->, draw] (root) edge node  [writeonarrow] {$\cdots$} (empty1);
        \path[->, draw] (root) edge node  [writeonarrow] {$\cdots$} (empty2);
       
    \end{tikzpicture}
    \center{(a)} 
    \bigskip
    \bigskip
    \bigskip
    \bigskip
    \bigskip
\end{minipage}
\begin{minipage}[b]{0.6\textwidth}
       \begin{tikzpicture}[thick,
                    node distance=2.2cm,
                    text height=1.5ex,
                    text depth=.25ex,
                    auto, 
                    squarednode/.style={rectangle, draw=red!60, fill=red!5, very thick, minimum size=5mm},
                    writeonarrow/.style = {pos=0.5,sloped,above}]
        \node [squarednode] (root) {$\mathcal{I}$};
        \node [squarednode, right of=root] (branch) {$\mathcal{I}$};
        \node [squarednode, above of=branch] (branch1) {$\mathcal{I}$};
        \node [squarednode, below of=branch] (branch2) {$\mathcal{I}$};
        \node [right of=branch] (empty) {};
        \node [left of=root] (emptyL) {};
        \node [below of=empty] (empty1) {};
        \node [above of=empty] (empty2) {};
        \node [below of=empty1] (empty0) {};
        \node [above of=empty2] (empty3) {};
        \node [right of=empty2] (empty5) {};
        \node [right of=empty1] (empty4) {};
        \node [right of=empty] (empty6) {};
        \node [below of=empty4] (empty7) {};
        \node [above of=empty5] (empty8) {};
        \node [below of=branch2] (nodedown) {};
        \node [above of=branch1] (nodeup) {};
        
        \path [->, draw] (emptyL) edge node [writeonarrow]  {$\rho$} (root);
        \path[->, draw] (root) edge node  [writeonarrow] {$\cdots$} (branch1);
        \path [->, draw] (root) edge node [writeonarrow]  {$\mathcal{I}_j[\rho]$} (branch);
        \path [->, draw] (root) edge node [writeonarrow]  {$\cdots$} (branch2);
        \path[->, draw] (branch) edge node  [writeonarrow] {$\cdots$} (empty6);
        \path[->, draw] (branch) edge node  [writeonarrow] {$\cdots$} (empty4);
        \path[->, draw] (branch) edge node [writeonarrow] {$\I_k \big[ \mathcal{I}_j [\rho] \big]$} (empty5);
        \path[->, draw] (branch2) edge node  [writeonarrow] {$\cdots$} (empty0);
        \path[->, draw] (branch2) edge node  [writeonarrow] {$\cdots$} (empty7);
        \path[->, draw] (branch1) edge node  [writeonarrow] {$\cdots$} (empty3);
        \path[->, draw] (branch1) edge node  [writeonarrow] {$\cdots$} (empty8);
        \path[->, draw] (branch1) edge node  [writeonarrow] {$\cdots$} (nodeup);
        \path[->, draw] (branch2) edge node  [writeonarrow] {$\cdots$} (nodedown);
       
    \end{tikzpicture}
    \center{(b)}
    
\end{minipage}
\caption{(a) Measuring apparatus. (b) Tree network for sequential uses of the same measuring apparatus.} \label{figure2}
\end{figure}
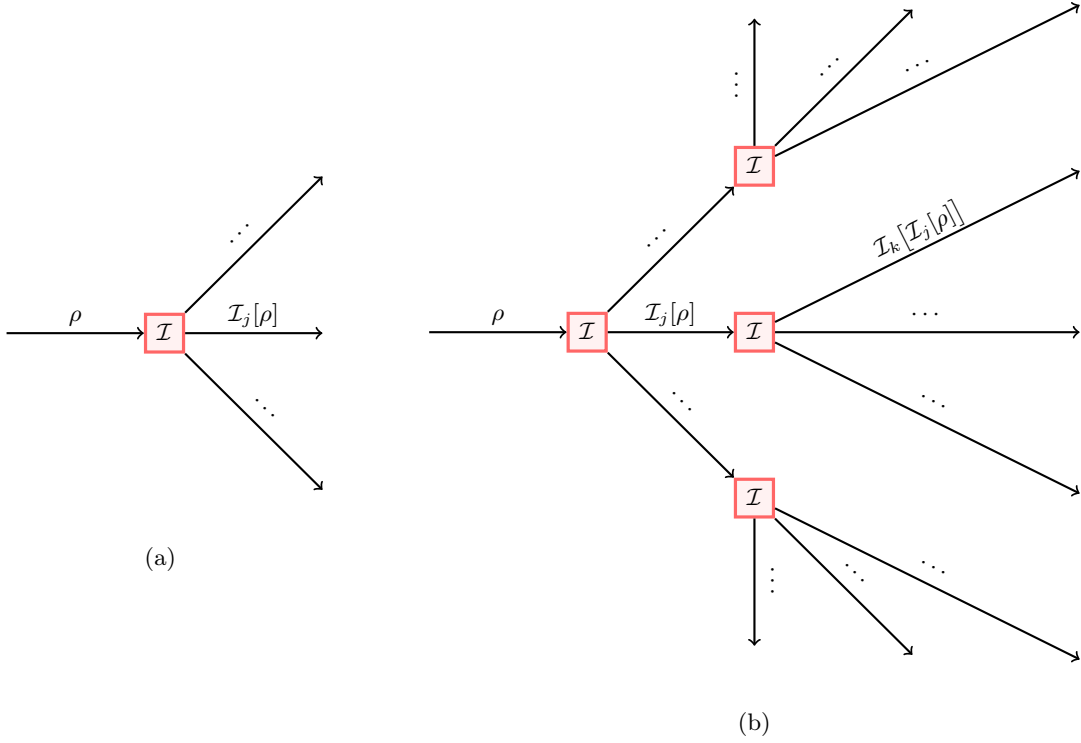

Consider an informationally incomplete measuring apparatus with $m$ outcomes, for which the mapping $\rho \to p_j$ is not injective. Physically, the distribution $\{p_j\}_{j=1}^m$ does not determine the density operator $\rho$ uniquely. The main question we address in this paper is whether the mapping $\rho \to p_{j_1 j_2 \ldots j_N}$ can become injective for some finite $N$. Physically, we study the problem whether $N$ uses of the same informationally incomplete measuring device can yield informationally complete statistics of multiindices. The following result provides a necessary and sufficient condition for the affirmative answer to the problem. 

\begin{theorem}\label{theorem}
$N$ sequential uses of an $m$-outcome measuring apparatus with instrument ${\cal I}$ provide informationally complete statistics if and only if
\begin{equation} \label{general-condition}
{\rm Span} \left( \left\{ \I^\dag_{j_1} \Big[ \ldots \big[ \I^\dag_{j_N}[I] \big] \Big] \right\}_{j_1,\ldots,j_N=1,\ldots,m} \right) = {\cal L}({\cal H}_d).
\end{equation} 
\end{theorem}
\begin{proof}
Rewriting the probability~\eqref{probability-multiindex} with the help of dual maps, we get

\noindent The statistics of outcomes $j_1 j_2 \ldots j_N$ is informationally complete if and only if all the effects $E_{j_1 j_2 \ldots j_N} := \I^\dag_{j_1} \Big[ \ldots \big[ \I^\dag_{j_N}[I] \big] \Big]$ span ${\cal L}({\cal H}_d)$, which completes the proof.
\end{proof}

Theorem~\ref{theorem} allows us to exclude the whole class of sharp (projective) measurements from consideration because they do not satisfy condition~\eqref{general-condition} as we show below. We refer to a measurement as sharp if the corresponding effects $E_j$ (given by formula~\eqref{effect-through-instrument}) are projectors, i.e., $E_j^2 = E_j$.

\begin{proposition} \label{proposition-projective}
Sequential uses of sharp measurements cannot provide an informationally complete statistics.
\end{proposition}
\begin{proof}
Every instrument that is compatible with a sharp POVM with effects $E_j$ has the form ${\cal I}_j [\rho] = \Phi_j [E_j \rho E_j]$, where $\{\Phi_j\}$ is a set of quantum channels~\cite{hayashi}. Therefore, ${\rm supp} \, \I^\dag_{j_1} \Big[ \ldots \big[ \I^\dag_{j_N}[I] \big] \Big] \subset {\rm supp} \, E_{j_1}$ and 
\begin{equation*}
{\rm Span} \left( \left\{ {\rm supp \,} E_{j_1} \right\}_{j_1=1,\ldots,m} \right) = \bigcup_{j} {\rm supp }\, E_j.
\end{equation*}

\noindent The relation $\bigcup_{j} {\rm supp} E_j = {\cal L}({\cal H}_d)$ holds true if and only if one of the projectors $E_{j'} = I$ and all other effects $E_{j}= O$, $j \neq j'$. If this is the case, then all operators $\I^\dag_{j_1} \Big[ \ldots \big[ \I^\dag_{j_N}[I] \big] \ldots \Big] = I$ due to unitality of quantum channels. Hence, the requirement~\eqref{general-condition} cannot be met by sharp measurements.
\end{proof}

Consider a special class of L\"{u}ders instruments ${\cal I}_j^{\rm L}$ that are fully determined by a reduced description in terms of the associated POVM effects $\{E_j\}_{j=1}^m$ as follows:
\begin{equation*}
{\cal I}_j^{\rm L} [\rho] = \sqrt{E_j} \rho \sqrt{E_j}. 
\end{equation*}

\begin{proposition} \label{proposition-Luders}
L\"{u}ders instrument with 2 outcomes cannot provide informationally complete statistics regardless of the number of uses.
\end{proposition}
\begin{proof}
Since $E_2 = I - E_1$, the operators $\sqrt{E_1}$ and $\sqrt{E_2}$ commute. Therefore, we have $\I^{{\rm L} \dag}_{j_1} \Big[ \ldots \big[ \I^{ {\rm L} \dag}_{j_N}[I] \big] \ldots \Big] = E_1^{n_1} E_2^{n_2}$, where $n_1 = \sum_{k=1}^m \delta_{1,j_k}$ and $n_2 = \sum_{k=1}^m \delta_{2,j_k}$. All the operators $E_1^{n_1} E_2^{n_2}$ are diagonal in the eigenbasis of the effect $E_1$, which implies 
\begin{equation*}
{\rm dim} \, {\rm Span} \, \left( \{E_1^{n_1} E_2^{n_2}\}_{n_1+n_2 = m}\right) \leqslant d < d^2 = {\rm dim} \, {\cal L}({\cal H}_d)
\end{equation*}

\noindent and impossibility to satisfy the requirement~\eqref{general-condition}.
\end{proof}

Propositions~\ref{proposition-projective} and~\ref{proposition-Luders} can be also considered as implications of the previously known results for commutative L\"{u}ders instruments~\cite{busch-1989}. To achieve the informational completeness and perform a precise quantum state tomography one has to implement several sequential measurements. The following result quantifies the lower bound on the number of measurements needed.

\begin{proposition} \label{proposition-minimal-number}
If $N$ sequential applications of a measuring apparatus for a $d$-level system with $m$ outcomes result in the informationally complete statistics, then $N \geqslant 2\log_m d$.
\end{proposition}
\begin{proof}
The measurement tree diagram of depth $N$ has $m^N$ possible outcomes. The informationally complete measurement has to have at least $d^2$ outcomes. Therefore, we readily obtain $m^N \geq d^2$ and $N \geq 2 \log_m d$.
\end{proof}

So far we have found a general condition~\eqref{general-condition} for informational completeness and revealed inability of projective measurements and L\"{u}ders measurements with 2 outcomes to provide informational completeness regardless of how many times they are used. The following section provides constructive examples of how the informational completeness emerges from sequential uses of the same informationally incomplete measurement.

\section{Qubit tomography via two sequential uses of a dichotomic measurement} \label{section-qubits}

We give two examples of a dichotomic ($m=2$) measuring apparatus for a qubit ($d=2$), which delivers informational completeness after the second application ($N=2$). Note that $N=2$ is the minimal number allowed by Proposition~\ref{proposition-minimal-number} in this case. Apparently, a single measuring apparatus is not informationally complete because $m<d^2$. In each example, the measurement tree has depth 2 and its 4 branches are enumerated by possible outcomes $\{j_1 j_2\}_{j_1,j_2=1,2}$. The probablity distribution $\{ p_{j_1 j_2} \}_{j_1,j_2=1,2}$ is informationally complete and enables reconstruction of the density operator $\rho$. 

To take into account experimental errors while estimating probabilities $\{ p_{j_1 j_2} \}_{j_1,j_2=1,2}$ by means of the corresponding relative frequencies, we quantify robustness of the proposed tomographic schemes via the conditional number of the associated Gram matrix
\begin{equation*}
G_{xy} =  \tr \left( E_x E_y\right), \quad x, y = \lbrace j_1,j_2 \rbrace.
\end{equation*}

\noindent The condition number is expressed through the eigenvalues $\{\lambda_i\}_{i=1}^{m^2}$ of the Gram matrix as follows 
\begin{equation*}
\Lambda = \frac{\max_i |\lambda_i|}{\min_i |\lambda_i|}.
\end{equation*}

\noindent The less the condition number, the more robust the density matrix reconstruction scheme to errors~\cite{filippov-2010,bogdanov-jetp-2010}. For instance, the most robust tomographic scheme with 4 outcomes for qubits is the symmetric informationally complete observable~\cite{caves-2004,fm-sic-2010}, for which the conditional number $\Lambda = 3$.

\begin{example} \label{example-qubit-Kraus-rank-2}
Define a one-parameter instrument with 2 outcomes
\begin{eqnarray} \label{operations-in-example-1}
&& \I_1 (\rho) = \frac{1}{2} H \left(  E  \rho E^\dagger  +  T \rho T^\dagger \right) H^{\dag}, \\
&&  \I_2 (\rho) = \frac{1}{2} V \left( E^\dagger \rho E  +  B  \rho B^\dagger \right) V^\dagger, \\
\end{eqnarray}

\noindent where the operators $E, T$, and $B$ are expressed through a real parameter $p \in [0,1]$ as follows:
\begin{equation*}
E = \begin{pmatrix}
 0 & \sqrt{p}\\
 0 & 0
 \end{pmatrix}, \quad  
T = \begin{pmatrix}
 \sqrt{1-p} & 0\\
 0 & 1 \\
 \end{pmatrix},  \quad
B = 
 \begin{pmatrix}
 1 & 0\\
 0 & \sqrt{1-p} \\
 \end{pmatrix},
\end{equation*}

\noindent and the unitary operators $H$ and $V$ read
\begin{equation} \label{H-V-unitaries}
H = \frac{1}{\sqrt{2}}
 \begin{pmatrix}
 1 & 1\\
 1 & -1
 \end{pmatrix}, \quad  V = \frac{1}{\sqrt{2}}
 \begin{pmatrix}
 1 & -i\\
 1 & i
 \end{pmatrix}.
\end{equation}

\noindent Since $H^{\dag}H = V^{\dag}V = I$, we readily obtain 
\begin{equation*}
     \I_1^\dagger (I) = \frac{1}{2}\begin{pmatrix}
     1 -p & 0\\
     0 & 1+p
     \end{pmatrix}, \quad \I_2^\dagger (I) = \frac{1}{2}\begin{pmatrix}
     1 + p & 0\\
     0 & 1-p
     \end{pmatrix}.
 \end{equation*}
 \noindent Direct calculation of the collective effects for the measurement tree of depth 2 yields
\begin{eqnarray*}
&&    E_{11} = \Id_1\left[ \Id_1 [I] \right] = \frac{1}{4} \left( I - \sigma_z  p -  \sigma_x p \sqrt{1-p}  \right), \\
&&     E_{12} = \Id_1\left[ \Id_2 [I] \right] = \frac{1}{4} \left( I - \sigma_z  p +  \sigma_x p \sqrt{1-p}  \right), \\
&&     E_{21} = \Id_2\left[ \Id_1 [I] \right] = \frac{1}{4} \left( I + \sigma_z  p -  \sigma_y p \sqrt{1-p}  \right), \\
&&      E_{22} = \Id_2\left[ \Id_2 [I] \right] = \frac{1}{4} \left( I + \sigma_z  p +  \sigma_y p \sqrt{1-p}  \right),\\
 \end{eqnarray*}
 
\noindent where $(\sigma_x,\sigma_y,\sigma_z)$ is the conventional set of Pauli operators. The obtained 4 effects $\{E_{j_1 j_2}\}_{j_1,j_2}$ are linearly independent self-adjoint operators in ${\cal L}({\cal H}_2)$ if $0<p<1$, therefore the mapping $\rho \to p_{j_1 j_2} = {\rm tr}(\rho E_{j_1 j_2})$ is injective and the statistics $\{ p_{j_1 j_2} \}_{j_1,j_2}$ is informationally complete if $0<p<1$. The optimal experimental implementation corresponds to the minimal condition number, which equals $\frac{27}{2}$ and is achieved at $p = \frac{2}{3}$.
 \end{example}
 
In the considered example, operations~\eqref{operations-in-example-1} have Kraus rank 2. In the following example, the operations defining an instrument have Kraus rank 1. 
 
\begin{example} \label{example-qubit-Kraus-rank-1}
\noindent Consider a one-parameter instrument with operations
\begin{eqnarray*}
&&    \I_1 [\rho] = \frac{1}{2} H \begin{pmatrix}
        \sqrt{1-p} & 0\\
        0 & \sqrt{1+p} \\
    \end{pmatrix}  \rho  \begin{pmatrix}
        \sqrt{1-p} & 0\\
        0 & \sqrt{1+p} \\
    \end{pmatrix} H^\dagger, \\
&&    \I_2 [\rho] = \frac{1}{2} V \begin{pmatrix}
        \sqrt{1+p} & 0\\
        0 & \sqrt{1-p} \\
    \end{pmatrix} \rho \begin{pmatrix}
        \sqrt{1+p} & 0\\
        0 & \sqrt{1-p} \\
    \end{pmatrix} V^\dagger, \\
\end{eqnarray*}
\noindent where a real parameter $p\in[0,1]$ and operators $H$ and $V$ are given by Eq.~\eqref{H-V-unitaries}. Some algebra yields
\begin{equation*}
    {\cal I}^\dagger_{1}[I] = \frac{1}{2}\left(I- p \sigma_z \right), \quad
    {\cal I}^\dagger_{2}[I] = \frac{1}{2}\left(I + p \sigma_z \right)
\end{equation*}

\noindent and the following collective effects
\begin{eqnarray*}
&&     E_{11} = \Id_1\left[ \Id_1 [I] \right] = \frac{1}{4} \left( I - \sigma_z  p -  \sigma_x p \sqrt{1-p^2}  \right), \\
&&     E_{12} = \Id_1\left[ \Id_2 [I] \right] = \frac{1}{4} \left( I - \sigma_z  p +  \sigma_x p \sqrt{1-p^2}  \right), \\
&&     E_{21} = \Id_2\left[ \Id_1 [I] \right] = \frac{1}{4} \left( I + \sigma_z  p -  \sigma_y p \sqrt{1-p^2}  \right), \\
&&      E_{22} = \Id_2\left[ \Id_2 [I] \right] = \frac{1}{4} \left( I + \sigma_z  p +  \sigma_y p \sqrt{1-p^2}  \right).\\
\end{eqnarray*}

\noindent The obtained effects are linearly independent self-adjoint operators in ${\cal L}({\cal H}_2)$ if $0<p<1$, which guarantees the informational completeness of the statistics $\{p_{j_1 j_2} = {\rm tr}(\rho E_{j_1 j_2})\}_{j_1,j_2}$. The optimal parameter $p = \frac{1}{\sqrt{2}}$ results in the minimal condition number $\Lambda = 8$. Comparing the condition numbers for examples~\ref{example-qubit-Kraus-rank-2} and \ref{example-qubit-Kraus-rank-1}, we conclude that the latter one is more robust to experimental errors.
\end{example}

\section{Two sequential measurements for $d$-level systems}\label{section-qudits}

In this section, we consider $d$-dimensional quantum systems and present a specific construction for the measuring apparatus with $m=d$ outcomes such that two sequential applications of this apparatus enable informational completeness of outcomes. 

Let $\{\ket{k}\}_{k=1}^d$ be an orthonormal basis in ${\cal H}_d$. Suppose a real parameter $p$ satisfies $-\frac{1}{d-1} \leqslant p \leqslant 1$, then the operator $\frac{1-p}{d} I + p \ketbra{k}{k}$ is positive semidefinite. Let $\{U_k\}_{k=1}^d$ be a set of unitary operators on ${\cal H}_d$, then the transformations
\begin{equation} \label{instrument-d-level}
    \I_k [\rho] = U_k \, \sqrt{\frac{1-p}{d} I + p \ketbra{k}{k}} \, \rho \, \sqrt{\frac{1-p}{d} I + p \ketbra{k}{k}} \, U_k^\dagger, \quad k = 1, \ldots, d,
\end{equation}

\noindent are completely positive and trace nonincreasing. Moreover, we have
\begin{equation*}
\sum_{k=1}^d {\cal I}_k^{\dag} [I] = \sum_{k=1}^d \left( \frac{1-p}{d} I + p \ketbra{k}{k} \right) = I,
\end{equation*}

\noindent which means that the map $k\to {\cal I}_k$ is a valid quantum instrument. The instrument~\eqref{instrument-d-level} is not informationally complete because the number of outcomes $m=d<d^2$. However, two sequential uses of this instrument lead to the following $d^2$ effects:
\begin{equation*}
E_{j_1 j_2} = \frac{1-p}{d} \left( \frac{1-p}{d} I + p \ketbra{j_1}{j_1} \right) + p \sqrt{\frac{1-p}{d} I + p \ketbra{j_1}{j_1}} \, U_{j_1}^{\dag} \ket{j_2}\bra{j_2} U_{j_1} \sqrt{\frac{1-p}{d} I + p \ketbra{j_1}{j_1}}.
\end{equation*}

Suppose $0<p<1$, then $\frac{1-p}{d} I + p \ketbra{j_1}{j_1}$ is a full rank operator for any $j_1$ and 
\begin{eqnarray*}
&& {\rm Span}\left( \left\{ \sqrt{\frac{1-p}{d} I + p \ketbra{j_1}{j_1}} \, U_{j_1}^{\dag} \ket{j_2}\bra{j_2} U_{j_1} \sqrt{ \frac{1-p}{d} I + p \ketbra{j_1}{j_1} } \right\}_{j_2 = 1}^d \right) \nonumber\\
&& = {\rm Span}\Big(\{U_{j_1}^{\dag} \ket{j_2}\bra{j_2} U_{j_1} \}_{j_2 = 1}^d \Big).
\end{eqnarray*}

\noindent Since $\sum_{j_2 = 1}^{d} E_{j_1 j_2} = \frac{1-p}{d} I + p \ketbra{j_1}{j_1}$, we have
\begin{equation*}
{\rm Span}\left(\left\{ \frac{1-p}{d} I + p \ketbra{j_1}{j_1} \right\} \bigcup \, \{U_{j_1}^{\dag} \ket{j_2}\bra{j_2} U_{j_1} \}_{j_2 = 1}^d \right) \subset {\rm Span}\left( \left\{ E_{j_1 j_2} \right\}_{j_2 = 1}^d \right)
\end{equation*}

\noindent and 

\begin{equation*}
{\rm Span}\left(\left\{ \frac{1-p}{d} I + p \ketbra{j_1}{j_1} \right\}_{j_1=1}^d \bigcup \, \{U_{j_1}^{\dag} \ket{j_2}\bra{j_2} U_{j_1} \}_{j_1,j_2 = 1}^d \right) \subset {\rm Span}\left( \left\{ E_{j_1 j_2} \right\}_{j_1,j_2 = 1}^d \right).
\end{equation*}

\noindent As ${\rm Span}\left(\left\{ \frac{1-p}{d} I + p \ketbra{j_1}{j_1} \right\}_{j_1 = 1}^d \right) = {\rm Span}\Big(\{ \ket{j_1}\bra{j_1} \}_{j_1=1}^d \Big)$, we finally get
\begin{equation*}
{\rm Span}\left(\{ \ket{j_1}\bra{j_1} \}_{j_1=1}^d \bigcup \, \{U_{j_1}^{\dag} \ket{j_2}\bra{j_2} U_{j_1} \}_{j_1,j_2 = 1}^d \right) \subset {\rm Span}\left( \left\{ E_{j_1 j_2} \right\}_{j_1,j_2 = 1}^d \right).
\end{equation*}

\noindent The known result in the theory of quantum state tomography for $d$-level systems is that there exists a set of unitary operators $\{I, U_1, \ldots, U_d\}$ such that 
\begin{equation*}
{\rm Span}\Big(\{ \ket{k}\bra{k} \}_{k=1}^d \bigcup \{ U_{k}^{\dag} \ket{l}\bra{l} U_{k} \}_{k,l=1}^d \Big) = {\cal L}({\cal H}_d),
\end{equation*}

\noindent see Refs.~\cite{wootters-1989,komisarski-2013}. Therefore, taking this set of unitary operators, we get ${\cal L}({\cal H}_d) \subset {\rm Span}\left( \left\{ E_{j_1 j_2} \right\}_{j_1,j_2 = 1}^d \right)$, which implies ${\rm Span}\left( \left\{ E_{j_1 j_2} \right\}_{j_1,j_2 = 1}^d \right) = {\cal L}({\cal H}_d)$. Using Theorem~\ref{theorem}, we conclude that two sequential applications of the instrument~\eqref{instrument-d-level} provide informationally complete statistics of outcomes $j_1 j_2$. To summarize, we have just proved the following result.

\begin{proposition} \label{proposition-qudits}
There exist unitary operators $U_1, \ldots, U_d$ such that two sequential uses of the instrument~\eqref{instrument-d-level} with $0<p<1$ provide informationally complete statistics of outcomes.
\end{proposition}

Let us illustrate our construction by the following example for qutrits ($d=3$). 

\begin{example}
The instrument~\eqref{instrument-d-level} has 3 outcomes. We choose unitary matrices $\{U_k\}_{k=1}^3$ in such a way that $\{U_k^{\dag}\}_{k=1}^3$ are transition matrices from the basis $\{\ket{k}\}_{k=1}^3$ to three more mutually unbiased bases, namely,
\begin{equation*}
    U_1^{\dag} = \frac{1}{\sqrt{3}}\begin{pmatrix}
    1 & 1 & w \\
    1 & w & w^2 \\
    1 & w^2 & w
    \end{pmatrix}, \quad U_2^{\dag} = \frac{1}{\sqrt{3}}\begin{pmatrix}
    1 & 1 & 1 \\
    w & w^2 & 1 \\
   w & 1 & w^2
    \end{pmatrix}, \quad U_3^{\dag} = \frac{1}{\sqrt{3}}\begin{pmatrix}
    1 & 1 & 1 \\
    w^2 & 1 & w \\
   w^2 & w & 1 
    \end{pmatrix},
\end{equation*}

\noindent where $w = \exp(i\frac{2\pi}{3})$~\cite{wootters-1989,MUB}. 

It is straightforward to verify that all the effects $\{E_{j_1 j_2}\}_{j_1,j_2 =1}^3$ are linearly independent if $0< p < 1$. The optimal experimental implementation corresponds to the minimal condition number of the Gram matrix, which equals $\Lambda \approx 17$ and is achieved at $p \approx 0.69$.
\end{example}

\section{Sequential dichotomic measurements for $n$-qubit systems} \label{section-n-qubits}

Consider a composite system composed of $n$ qubits, i.e., $d=2^n$, and a dichotomic informationally incomplete measuring apparatus ($m=2$). By Proposition~\ref{proposition-minimal-number}, in order to provide an informationally complete statistics the measuring device should be used $N \geqslant 2n$ times. The minimal depth $N=2n$ of the corresponding binary measurement tree is sufficient, indeed, as the following construction shows. 

Suppose $\left\{ {\cal I}_1, {\cal I}_2 \right\}$ is a dichotomic instrument for a single qubit such that ${\rm Span} \left( \left\{ {\cal I}_{j_1}^{\dag} \left[{\cal I}_{j_2}^{\dag}[I] \right] \right\}_{j_1,j_2=1,2} \right) = {\cal L}({\cal H}_2)$, for instance, an instrument from Example~\ref{example-qubit-Kraus-rank-2} or Example~\ref{example-qubit-Kraus-rank-1}. Consider the following instrument for $n$ qubits:
\begin{equation} \label{instrument-with-shift}
\widetilde{\cal I}_k [\rho] = U_{\rm shift} \left( {\cal I}_k \otimes {\rm Id} \otimes \ldots \otimes {\rm Id} [\rho] \right) U_{\rm shift}^{\dag}, \quad k=1,2,
\end{equation}

\noindent where the identity transformation ${\rm Id}: {\cal L}({\cal H}_2) \to {\cal L}({\cal H}_2)$ appears $n-1$ times,
\begin{equation*}
U_{\rm shift} = \sum_{i_1,\ldots,i_n = 1}^{2}
\ket{i_n}\bra{i_1} \otimes \ket{i_1}\bra{i_2} \otimes
\ket{i_2}\bra{i_3} \otimes \ldots \otimes \ket{i_{n-1}}\bra{i_n}
\end{equation*}

\noindent is a unitary operator shifting the particles (used in Ref.~\cite{lvgf-2019}), and $\{\ket{1},\ket{2}\}$ is an orthonormal basis for a single qubit. It is not hard to see that $n$ sequential applications of the instrument~\eqref{instrument-with-shift} lead to the effects
\begin{equation*}
E_{j_1 j_2 \ldots j_{2n}} = {\cal I}_{j_1}^{\dag} \left[{\cal I}_{j_{n+1}}^{\dag}[I] \right] \otimes {\cal I}_{j_2}^{\dag} \left[{\cal I}_{j_{n+2}}^{\dag}[I] \right] \otimes \ldots \otimes {\cal I}_{j_n}^{\dag} \left[{\cal I}_{j_{2n}}^{\dag}[I] \right],
\end{equation*}

\noindent where $I: {\cal H}_2 \to {\cal H}_2$. Obviously, ${\rm Span}\left( \{E_{j_1 j_2 \ldots j_{2n}} \}_{j_1,j_2,\ldots,j_{2n} = 1,2} \right) = {\cal L}({\cal H}_{2^n})$ because ${\rm Span} \left( \left\{ {\cal I}_{j_1}^{\dag} \left[{\cal I}_{j_2}^{\dag}[I] \right] \right\}_{j_1,j_2=1,2} \right) = {\cal L}({\cal H}_2)$. Consequently, by Theorem~\ref{theorem} we deduce the informational completeness of $2n$ uses of the instrument~\eqref{instrument-with-shift}.

The results of this section naturally generalize to more complicated composite systems consisting of $d$-level systems. The  measurement tree has depth $N'n$, where $n$ is the number of $d$-level systems under study and $N'$ is the number of sequential measurements sufficient for tomography of a single $d$-level system.

\section{Conclusions} \label{section-conclusions}

We considered a non-destructive measuring apparatus that leaves the system state available for further analysis after the measurement. Even if the measurement is informationally incomplete, it may happen that $N$ sequential uses of the same apparatus do provide informationally complete statistics. We fully characterized those measuring apparatuses in Theorem~\ref{theorem} by using a dual map to the corresponding quantum instrument. In Propositions~\ref{proposition-projective} and~\ref{proposition-Luders} we showed that projective measurements with any number of outcomes and L\"{u}ders measurements with 2 outcomes fail in satisfying the requirement of Theorem~\ref{theorem} and, therefore, cannot provide informational completeness regardless of the number of uses. In Proposition~\ref{proposition-minimal-number} we found a lower bound on how many times the measuring apparatus is to be used to make the informational completeness feasible. This lower bound is shown to be achievable for dichotomic qubit measurements in Examples~\ref{example-qubit-Kraus-rank-2} and~\ref{example-qubit-Kraus-rank-1} as well as for $d$-outcome measurements for $d$-level systems (Proposition~\ref{proposition-qudits}). The obtained results were generalized to composite systems in Section~\ref{section-n-qubits}.

\section{Acknowledgements}

The authors thank Teiko Heinosaari for useful comments and bringing Refs.~\cite{haapasalo-2016,busch-1989,komisarski-2013} to our attention.


\end{document}